\documentclass[%
 reprint,
 amsmath,amssymb,
 aps,
 nofootinbib
]{revtex4-1}

\usepackage{amsthm}
\usepackage[dvipdfmx]{graphicx}
\usepackage[dvipdfmx]{xcolor}
\usepackage{dcolumn}
\usepackage{bm}
\usepackage{braket}
\usepackage{amsmath}
\usepackage{amssymb}
\usepackage{amsthm}
\usepackage{enumerate}
\usepackage{booktabs}
\usepackage{txfonts}
\usepackage[update,prepend]{epstopdf}
\usepackage[english]{babel}
\usepackage{listings}
\usepackage{qcircuit}
\usepackage[subrefformat=parens]{subcaption}
\usepackage{caption}
\usepackage{multirow}
\usepackage{algorithmic}
\usepackage{algorithm}
\usepackage{mathtools}
\usepackage{here}
\newcommand{\proc}[1]{\textup{\textsf{#1}}}

\newtheorem{theorem}{Theorem}

\newtheorem{definition}{Definition}

\makeatletter
\def\hlinewd#1{%
	\noalign{\ifnum0=`}\fi\hrule \@height #1 %
	\futurelet\reserved@a\@xhline}
\makeatother

\begin{document}

\preprint{APS/123-QED}

\title{Quantum algorithm for calculating risk contributions in a credit portfolio}

\author{Koichi Miyamoto}
\email{koichi.miyamoto@qiqb.osaka-u.ac.jp}
\affiliation{Center for Quantum Information and Quantum Biology, Osaka University \\ 1-3 Machikaneyama, Toyonaka, Osaka, 560-8531, Japan}

\date{\today}

\begin{abstract}

Finance is one of the promising field for industrial application of quantum computing.
In particular, quantum algorithms for calculation of risk measures such as the value at risk and the conditional value at risk of a credit portfolio have been proposed.
In this paper, we focus on another problem in credit risk management, calculation of risk contributions, which quantify the concentration of the risk on subgroups in the portfolio.
Based on the recent quantum algorithm for simultaneous estimation of multiple expected values, we propose the method for credit risk contribution calculation.
We also evaluate the query complexity of the proposed method and see that it scales as $\widetilde{O}\left(\sqrt{N_{\rm gr}}/\epsilon\right)$ on the subgroup number $N_{\rm gr}$ and the accuracy $\epsilon$, in contrast with the classical method with $\widetilde{O}\left(\log(N_{\rm gr})/\epsilon^2\right)$ complexity.
This means that, for calculation of risk contributions of finely divided subgroups, the advantage of the quantum method is reduced compared with risk measure calculation for the entire portfolio.
Nevertheless, the quantum method can be advantageous in high-accuracy calculation, and in fact yield less complexity than the classical method in some practically plausible setting. 

\end{abstract}

\pacs{Valid PACS appear here}
                              
\maketitle

\section{\label{sec:intro}Introduction}

Following the recent advance of quantum computing technology, people are now looking for industrial applications.
Finance is one of promising fields (see \cite{Orus,Egger,Bouland,Herman} as comprehensive reviews).
In particular, credit portfolio risk measurement is a problem for which applications of quantum algorithms are actively investigated.
Every bank has a credit portfolio that consists of loans it has issued, and it is exposed to the credit risk, that is, the risk to incur the loss by defaults of obligors.
In order to monitor such a risk, banks calculate risk measures that quantify the amount of the risk, such as the value at risk (VaR), the percentile (e.g. 99\%) of the loss distribution, and the conditional VaR (CVaR), the conditional expectation of the loss given that it exceeds the VaR (see \cite{Fischer} as a recent review).
They are often evaluated by the Monte Carlo method \cite{Glasserman}, in which defaults are randomly generated by some mathematical model and many sample values of the loss are taken.
Generating many samples, whose number is typically of order $10^6$, for a large credit portfolio, which can contain millions of obligors for major banks, is of high computational cost.
On the other hand, there are some quantum algorithms for Monte Carlo integration \cite{Montanaro,Suzuki,Herbert} based on quantum amplitude estimation \cite{Suzuki,Brassard,Aaronson,Nakaji,Giurgica-Tiron,Grinko,Tanaka,Uno,Wang}. 
Although the classical Monte Carlo integration has sample complexity scaling as $O(\epsilon^{-2})$ on $\epsilon$, the error tolerance for the integral, query complexity in the quantum counterparts scales as $O(\epsilon^{-1})$, which is often referred to as quantum quadratic speedup.
Therefore, quantum Monte Carlo integration has been applied to credit portfolio risk measurement in previous studies \cite{Egger2,Miyamoto,Kaneko}, and is expected to provide speedup and sophistication for credit risk management in banks in the future.

In this paper, we focus on another problem in credit risk management, that is, calculation of risk contributions.
VaR and CVaR indicate the amount of the risk in the entire portfolio, but we sometimes want to know how each subportfolio contributes to the risk measures.
In other words, we want to quantitatively evaluate the concentration of the risk to subgroups of obligors.
If a bank has such measures, they will help it to analyze and disperse the risk.
For example, it can notice that the risk concentrates on obligors in some specific industrial sector or some specific region, and that some business division in it are taking the risk too much.
In fact, such measures called risk contributions have been defined and how to calculate them has been studied \cite{Fischer,Litterman,Tasche,Kurth,Kalkbrener,Kalkbrener2,Glasserman2,Martin,Muromachi,Muromachi2}.
They are written as conditional expected values, and therefore calculating them by Monte Carlo is more costly than the entire risk measures, since only a small fraction of samples matches the conditions and we need to generate more samples.

Then, this paper aims at quantum speedup of credit risk contribution calculation.
Note that original quantum algorithms for Monte Carlo integration \cite{Montanaro,Suzuki,Herbert} output an estimate of a single expected value, and that sequentially applying such an algorithm to calculation of risk contributions of $N_{\rm gr}$ obligor groups, which leads to $O(N_{\rm gr}/\epsilon)$ complexity, is not efficient.
Instead, we resort to the recently proposed quantum algorithm for simultaneous calculation of expected values of multiple random variables \cite{Cornelissen} (see also \cite{Huggins}).
This algorithm outputs $d$ expected values with accuracy $\epsilon$, making $\widetilde{O}\left(\sqrt{d}/\epsilon\right)$ queries\footnote{$\widetilde{O}(\cdot)$ denotes $O(\cdot)$ in the big-O notation with some logarithmic factors hidden.} to two oracles, of which one generates a state that encodes the probability distribution in amplitudes and the other computes the random variables. 
Presenting how to construct these oracles concretely, we show that this quantum algorithm can be in fact applied to risk contribution calculation.
We then see that the number of queries to building-block operations in the proposed method scales as $\widetilde{O}\left(\sqrt{N_{\rm gr}}/\epsilon\right)$ on $N_{\rm gr}$ and $\epsilon$.
This means quantum quadratic speedup with respect to $\epsilon$, but not with respect to $N_{\rm gr}$, since the classical Monte Carlo method has $\widetilde{O}\left(\log(N_{\rm gr})/\epsilon^2\right)$ complexity.
In general, the more finely divided obligor groups we set, the more the advantage of the quantum method is reduced.
Nevertheless, as we will see later, in some practically plausible problem setting, the proposed method seems to be advantageous against the classical method in terms of query complexity.

The rest of this paper is organized as follows.
In Section \ref{sec:credit}, we review the Merton model \cite{Merton}, a widely-used mathematical model for credit portfolio risk measurement, and introduce risk measures such as VaR and CVaR and risk contributions.
Section \ref{sec:main} is the main part of this paper.
In this section, as preparations, we present some oracles used as building blocks in the proposed method such as arithmetic operations, controlled rotation and generation of a state corresponding to a standard normal random variable, and introduce quantum multiple expected value estimation \cite{Cornelissen} and fixed-point quantum amplitude amplification (QAA) \cite{Yoder}, which is another important base quantum algorithm.
Then, we present the proposed quantum method as Theorem \ref{th:main}, the main result of this work, and discuss its complexity and comparison with the classical method.
Section \ref{sec:sum} summarizes this paper.

\subsection{Notations}

Here, we summarize some notations we later use.
$\mathbb{R}_+$ denotes the set of all positive real number: $\mathbb{R}_+:=\{x\in\mathbb{R} \ | \ x>0\}$.
For $n\in\mathbb{N}$, we define $[n]:=\{1,...,n\}$.
For $\epsilon\in\mathbb{R}_+$, we say that $x\in\mathbb{R}$ is a $\epsilon$-approximation of $y\in\mathbb{R}$ if $|x-y|\le\epsilon$.
$1_C$ denotes the indicator function, which takes 1 if the condition $C$ is satisfied and 0 otherwise.
$\delta_{a,b}$ denotes the Kronecker delta: for integers $a$ and $b$, $\delta_{a,b}=1$ if $a=b$ and 0 otherwise. 

\section{\label{sec:credit} Credit risk measures and risk contributions}

In this section, we introduce some risk measures and risk contributions for a credit portfolio.

\subsection{\label{sec:model} Credit risk model}

First, we introduce a mathematical model for credit portfolio risk measurement, the Merton model \cite{Merton}, which is widely used in practical business and considered in this paper.
Let us consider a credit portfolio consisting of $N_{\rm obl}$ obligors, whose exposures\footnote{An exposure of an obligor is the amount of loss that occurs when it defaults. It is expressed as the product of the loan amount and the loss given default (LGD), the fraction of the amount which is not recovered. Although the LGD is sometimes modeled as a random variable, it is treated as a constant and therefore so is the exposure in this paper.} are $e_1,...,e_{N_{\rm obl}}\in\mathbb{R}_+$.
Then, in the Merton model, we model the loss $L$ in the portfolio as follows, using independent random variables $X_0,X_1,...,X_{N_{\rm obl}}$ that follows the standard normal distribution:
\begin{equation}
	L(X_0,X_1,...,X_{N_{\rm obl}})=\sum_{k=1}^{N_{\rm obl}} e_kY_k(X_0,X_k). \label{eq:loss}
\end{equation}
Here, for each $k\in[N_{\rm obl}]$, $Y_k$ is defined as
\begin{equation}
	Y_k(X_0,X_k):=1_{Z_k(X_0,X_k)<z_k}, \label{eq:Yk}
\end{equation}
and $Z_k$ is defined as
\begin{equation}
	Z_k(X_0,X_k):=a_k X_0 + \sqrt{1-a_k^2}X_k, \label{eq:Zk}
\end{equation}
with constants $z_k\in\mathbb{R}$ and $a_k\in(0,1)$.
$Y_k=1$ and $Y_k=0$ mean that the $k$th obligor defaults and not, respectively.
$Z_k$ is called the value of the firm of the $k$th obligor, and it defaults if $Z_k$ falls below some threshold $z_k$.
It is set according to the default probability of the obligor, which is usually set based on some credit rating model.
Note that $X_0$ is common for all the obligors but $X_1,...,X_{N_{\rm obl}}$ affect only the first, ..., $N_{\rm obl}$th obligors, respectively.
$X_0$ is called a systematic risk factor, which reflects the situation of macro economy, and $X_1,...,X_{N_{\rm obl}}$ are called idiosyncratic risk factors, which reflect the matters unique to the credit of the individual obligor\footnote{Although the model with multiple systematic risk factors are often used, we consider the single-factor case for simplicity in this paper. Extending the discussion in this paper to the multi-factor case is straightforward.}.
$Z_k$ is the linear combination of $X_0$ and $X_k$ as (\ref{eq:Zk}), and follows the standard normal distribution too. 
The coefficients $a_1,...,a_{N_{\rm obl}}$ control correlations between the values of the firms: ${\rm Cor}(Z_k,Z_{k^\prime})=\alpha_k\alpha_{k^\prime}$ for different $k,k^\prime\in[N_{\rm obl}]$.
Therefore, setting these coefficients close to 1 leads to strong correlations, which make simultaneous defaults of many obligors more probable and fatten the tail of the loss distribution.
Also note that, under the condition that $X_0$ takes a value $x_0\in\mathbb{R}$, $Y_k$ follows a Bernoulli distribution with the probability that $Y_k=1$ being
\begin{equation}
	P^{\rm def}_k(x_0):=\Phi_{\rm SN}\left(\frac{z_k-a_kx_0}{\sqrt{1-a_k^2}}\right),
\end{equation}
where $\Phi_{\rm SN}$ is the cumulative distribution function (CDF) of the standard normal distribution.

\subsection{\label{sec:riskMea} Risk measures}

In order to quantitatively assess the risk in a portfolio, some risk measures have been proposed.
Most widely used one is the value at risk (VaR).
Given $\alpha\in(0,1)$, which is typically set so that $1-\alpha=0.99$ (99\%) or 0.999 (99.9\%), we define the $100(1-\alpha)$\% VaR as
\begin{equation}
	V_\alpha := \inf\{x\in\mathbb{R} \ | \ {\rm Pr}(L(X_0,X_1,...,X_{N_{\rm obl}})\ge x)\le \alpha \}.
	\label{eq:VaR}
\end{equation}
This means that the probability that the loss larger than $100(1-\alpha)$\% VaR occurs is at most $\alpha$.

Another widely used measure is the conditional value at risk (CVaR), which is also known as the expected shortfall.
Given $v\in\mathbb{R}_+$, we define
\begin{equation}
	C_v := \mathbb{E}_{\rm Mer}[L(X_0,X_1,...,X_{N_{\rm obl}}) \ | \ L(X_0,X_1,...,X_{N_{\rm obl}})\ge v], \label{eq:CVaR}
\end{equation}
where $\mathbb{E}_{\rm Mer}[\cdot]$ denotes the (conditional) expected value with respect to randomness of $X_0,X_1,...,X_{N_{\rm obl}}$.
Then, $100(1-\alpha)$\% CVaR is defined as $C_{V_\alpha}$, that is, (\ref{eq:CVaR}) with $v=V_\alpha$.
In practice, we first obtain some estimation $v$ of the VaR, and then calculate $C_v$ as an estimation of the CVaR.

\subsection{\label{sec:riskCont} Risk contributions}

In practice, calculating only risk measures for the entire portfolio is not sufficient.
Sometimes, we need to decompose the risk measures to contributions from the subgroups in the portfolio, in order to, for example, analyze concentration of the risk.
There are some studies on how to define and calculate such risk contributions \cite{Litterman,Tasche,Kurth,Kalkbrener,Kalkbrener2,Glasserman2,Martin,Muromachi,Muromachi2}, which we now outline.

First, we present how to represent the subgroups.
We assume that the obligors in the portfolio are divided into $N_{\rm gr}$ groups, in which the numbers of obligors are $n_1,...,n_{N_{\rm gr}}\in\mathbb{N}$.
Without loss of generality, we can assign small indexes to the obligors in the group with a small index.
That is, for each $K\in[N_{\rm gr}]$, we assume that the $K$th group contains $\kappa^K_1$th, ..., $\kappa^K_{n_K}$th obligors, where
\begin{equation}
	\kappa^K_k :=
	\begin{cases}
		k & ; \ {\rm for} \ K=1,k\in[n_1] \\
		\sum_{K^\prime=1}^{K-1} n_{K^\prime} + k & ; \ {\rm for} \ K\in\{2,...,N_{\rm gr}\},k\in[n_K]
	\end{cases}.
\end{equation}
Under this indexing, the loss in the $K$th group is given as
\begin{equation}
	L_K(X_0,X_{\kappa^K_1},...,X_{\kappa^K_{n_K}}):=\sum_{k=\kappa^K_1}^{\kappa^K_{n_K}} e_kY_k(X_0,X_k). \label{eq:LK}
\end{equation}
We set the group according to how finely we want to analyze the risk.
For example, if we want risk contributions of all individual obligors, we set $N_{\rm obl}$ groups, each of which contains only one obligor.
If not, we try appropriate grouping according to our purpose.
For example, we can group obligors by their industrial and/or regional sectors, in order to monitor which sector the risk concentrate on and how large it is.

Now, we define the risk contributions as follows.
With respect to the VaR, the risk contribution of the $K$th group is
\begin{equation}
	V^{K}_\alpha := \mathbb{E}_{\rm Mer}\left[L_K \ \middle| \ L= V_\alpha\right], \label{eq:VaRCont}
\end{equation}
and that with respect to the CVaR is
\begin{equation}
	C^{K}_v:= \mathbb{E}_{\rm Mer}\left[L_K \ \middle| \ L\ge v\right]. \label{eq:CVaRCont}
\end{equation}
Note that the sum of risk contributions over the groups is equal to the risk measure for the entire portfolio: $\sum_{K=1}^{N_{\rm gr}}V^{K}_\alpha=V_\alpha$ and $\sum_{K=1}^{N_{\rm gr}}C^{K}_v=C_v$, which makes these definitions of risk contributions plausible.

Also note that Monte Carlo estimation of risk contributions is harder than that of the risk measure for the entire portfolio.
This is because of the definitions of risk contributions as conditional expected values.
When we randomly generate samples of the random variable set $(X_0,X_1,...,X_{N_{\rm obl}})$ or $(X_0,Y_1,...,Y_{N_{\rm obl}})$, only the small fraction of them satisfy the condition $L=V_\alpha$ or $L\ge v$.
Although we can nonetheless estimate CVaR contributions with the probability that $L\ge v$ being not too small (say 0.01), estimating VaR contributions is much harder since the probability that the loss takes a specific value $V_\alpha$ is much smaller.
Even if we resort to quantum algorithms, estimating VaR contributions is hard, since the probability that the condition for the conditional expected value is satisfied affects the complexity, as we will see below.
In light of this, we hereafter focus on CVaR contributions, only referring to some existing studies on VaR contribution calculation, such as the Monte Carlo method combined with importance sampling \cite{Glasserman2} and semi-analytical methods based on saddle-point approximation \cite{Martin,Muromachi,Muromachi2}.

\section{\label{sec:main} Quantum algorithm for measuring risk contributions in a credit portfolio}

\subsection{Fixed-point binary representation by qubits}

Before we present the quantum algorithm for CVaR contribution calculation, we need some preparations.
First of all, we now present the current setting for numerical calculation on a quantum computer.
In this paper, we consider systems consisting of some qubits, and use the bit string on a quantum register as the real number in fixed-point binary representation.
More strictly, in this paper, a quantum register, or simply a register, means a system consisting of $N_{\rm dig}$ qubits, where $N_{\rm dig}$ is sufficiently large, and, for every $x\in\mathbb{R}$, $\ket{x}$ denotes a computational basis state on such a register with a bit string corresponding to ${\rm argmin}_{a\in \mathcal{G}}|x-a|$, where $\mathcal{G}$ is some set of fixed-point binary numbers with $N_{\rm dig}$ digits.
We assume that $\mathcal{G}$ covers the sufficiently wide range on the real axis with the sufficiently small interval, and that therefore approximating a real numbers with an element of $G$ yields only a small error, which we hereafter neglect.
Besides, for $\vec{v}=(v_1,...,v_d)^T\in\mathbb{R}^d$, $\ket{\vec{v}}$ denotes a state on $d$-registers system such that $\ket{\vec{v}}:=\ket{v_1}\cdots\ket{v_d}$. 

\subsection{Building-block oracles}

Next, we introduce some oracles, which are used as parts of the algorithm.

\begin{definition}
	We call the following oracles on a multi-register system arithmetic circuits.
	
	\begin{itemize}
		\item Addition $U_{\rm add}$: for any $x,y\in\mathbb{R}$, $U_{\rm add}\ket{x}\ket{y}=\ket{x}\ket{x+y}$
		\item Multiplication $U_{\rm mul}$: for any $x,y\in\mathbb{R}$, $U_{\rm add}\ket{x}\ket{y}\ket{0}=\ket{x}\ket{y}\ket{xy}$
		\item Comparison $U_{\rm comp}$: for any $x,y\in\mathbb{R}$, $U_{\rm comp}\ket{x}\ket{y}\ket{0}=\ket{x}\ket{y}(1_{x\ge y}\ket{1}+1_{x< y}\ket{0})$
		\item Standard normal CDF $U_{\Phi_{\rm SN}}$: for any $x\in\mathbb{R}$, $U_{\Phi_{\rm SN}}\ket{x}\ket{0}=\ket{x}\ket{\Phi_{\rm SN}(x)}$
		\item Arccos and square root $U_{\rm acsr}$: for any $x\in[0,1]$, $U_{\rm acsr}\ket{x}\ket{0}=\ket{x}\ket{\arccos(\sqrt{x})}$
	\end{itemize}
\end{definition}

In fact, many proposals on circuit implementations for addition and multiplication have been made \cite{Vedral,Beckman,Draper,Cuccaro,Takahashi,Draper2,Alvarez-Sanchez,Takahashi2,Takahashi3,Babu,Jayashree,Munoz-Coreas}.
Comparison of $x$ and $y$ is virtually a subtraction $y-x$, since, when we use the 2's complement method to represent a negative number, the top bit of $y-x$ is 1 if $y-x\le 0$ and $0$ otherwise.
Also note that subtraction can be done as addition in the 2's complement method.
For the CDF of the standard normal distribution $\Phi_{\rm SN}$, or, equivalently, the error function ${\rm erf}(x)=2\Phi_{\rm SN}(\sqrt{2}x)-1$, many accurate approximations with elementary functions have been proposed: for example, ${\rm erf}(x)\approx 1-1/(1+a_1x+\cdots a_6x^6)^{16}$ with $a_1 = 0.0705230784,a_2 = 0.0422820123,a_3 = 0.0092705272,a_4 = 0.0001520143,a_5 = 0.0002765672,a_6 = 0.0000430638$ \cite{Abramowitz}.
Therefore, $U_{\Phi_{\rm SN}}$ can be approximately implemented with circuits for addition, multiplication and division (note that division circuits have also been proposed \cite{Khosropour,Jamal,Dibbo,Thapliyal}).
For $U_{\rm acsr}$, we can use implementation of inverse trigonometric functions in \cite{Haner}, and that of square root in \cite{Munoz-Coreas2}.

\begin{definition}
	We call the following oracles $U_{\rm CRY}$ on a system consisting of a quantum register and a qubit the angle-controlled Y rotation:
	\begin{equation}
		U_{\rm CRY}:=\sum_{\theta\in \mathcal{G}} \ket{\theta}\bra{\theta}\otimes R_Y(\theta)
	\end{equation}
	Here, $R_Y(\theta):=\begin{pmatrix}	\cos\theta & \sin\theta \\ -\sin\theta & \cos\theta \end{pmatrix}$ for any $\varphi\in\mathbb{R}$.
\end{definition}
That is, $U_{\rm CRY}$ is the rotation on the Bloch sphere around the $Y$-axis, whose angle $\theta$ is specified by another register.
These gates can be implemented as a series of multi-controlled rotation gates with fixed angles \cite{Egger2}.

\begin{definition}
We call the following oracle on a quantum register the SN state generation oracle:
\begin{equation}
	U^{\rm SN}\ket{0}=\ket{\rm SN}:=\sum_{i=1}^{N_{\rm SN}}\sqrt{p^{\rm SN}_i}\ket{x^{\rm SN}_{i}}, \label{eq:OraSN}
\end{equation}
where $N_{\rm SN}$ is an integer not less than 2,
\begin{equation}
	p^{\rm SN}_i:=
	\begin{dcases}
		\int_{-\infty}^{\frac{x^{\rm SN}_1+x^{\rm SN}_2}{2}} \phi_{\rm SN}(x)dx & ; \ {\rm for} \ i=1 \\
		\int_{\frac{x^{\rm SN}_{i-1}+x^{\rm SN}_i}{2}}^{\frac{x^{\rm SN}_{i}+x^{\rm SN}_{i+1}}{2}} \phi_{\rm SN}(x)dx & ; \ {\rm for} \ i\in\{2,...,N_{\rm SN}-1\} \\
		\int^{\infty}_{\frac{x^{\rm SN}_{N_{\rm SN}-1}+x^{\rm SN}_{N_{\rm SN}}}{2}} \phi_{\rm SN}(x)dx & ; \ {\rm for} \ i=N_{\rm SN}
	\end{dcases},
\end{equation}
$\phi_{\rm SN}$ is the probability density function for the standard normal distribution, and, for $i\in[N_{\rm SN}]$, $x^{\rm SN}_i:=-D+2D(i-1)/(N_{\rm SN}-1)$ with $D\in\mathbb{R}_+$.
\end{definition}

This can be interpret as an oracle to generate a state in which the standard normal distribution is approximately encoded in the amplitudes.
That is, with sufficiently large $N_{\rm SN}$ and $D$, the random variable $X_{\rm DiscSN}$ that takes $x^{\rm SN}_1,...,x^{\rm SN}_{N_{\rm SN}}$ with probabilities $p^{\rm SN}_1,...,p^{\rm SN}_{N_{\rm SN}}$, respectively, can be viewed as a discretized approximation of a standard normal random variable, and we obtain $x^{\rm SN}_i$ with probability $p_i$ when we measure $\ket{\rm SN}$.
We hereafter use $X_{\rm DiscSN}$ as if it followed the standard normal distribution and neglect any errors caused by this approximation.

Some implementations for this type of oracle have been proposed.
Originally, the implementation as a series of arithmetic circuits and controlled rotation was proposed in \cite{Grover}, and some extensions and modifications on this have been also proposed recently \cite{Sanders,Kaneko2,Marin-Sanchez}.
In another direction, some state preparation methods based on variational algorithms with parametric quantum circuits have been considered \cite{Zhu,Zoufal,Chakrabarti,Nakaji2,Alcazar}.

\subsection{\label{sec:qMultiExp} Quantum algorithm for estimating multiple expected values}

We next explain the quantum algorithm for simultaneous estimation of expected values of multiple random variables, which is the core of the proposed CVaR contribution calculation method.
There are some recent proposals on such a quantum algorithm \cite{Huggins,Cornelissen}.
In this paper, we use the algorithm in \cite{Cornelissen}.
Formally, we have the following theorem.

\begin{theorem}[Theorem 3.4 in \cite{Cornelissen}]
	
	Let $(\Omega,2^\Omega,\mathbb{P})$ be a probability space where the sample space $\Omega$ is finite and elements $\omega$ in $\Omega$ are associated with mutually orthogonal states $\ket{\omega}$ on a quantum register. 
	Let $\vec{\xi}$ be a $\mathbb{R}^d$-valued random variable on $(\Omega,2^\Omega,\mathbb{P})$, and suppose that $\vec{\xi}$ has a mean $\vec{\mu}=(\mu_1,...,\mu_d)^T\in\mathbb{R}^d$ and a covariance matrix $\Sigma\in\mathbb{R}^{d\times d}$.
	Suppose that we have accesses to an oracle $U_{\mathbb{P}}$ on a quantum register, which acts as
	\begin{equation}
		U_{\mathbb{P}}\ket{0}=\sum_{\omega\in\Omega} \sqrt{\mathbb{P}(\omega)}\ket{\omega}, \label{eq:UP}
	\end{equation} 
	and an oracle $U_{\vec{\xi}}$ on a $(d+1)$-registers system, which acts as
	\begin{equation}
		U_{\vec{\xi}}\ket{\omega}\ket{0}^{\otimes d}=U_{\vec{\xi}}\ket{\omega}\ket{\vec{\xi}(\omega)} \label{eq:Uxi}
	\end{equation}
	for any $\omega\in\Omega$.
	Then, for given $\delta\in(0,1)$ and $n\in\mathbb{N}$ such that $n\ge\log(d/\delta)$, there is a quantum algorithm $\proc{QEstimator}_d(\vec{\xi},n,\delta)$ that outputs an estimate $\vec{\mu}^\prime=(\mu^\prime_1,...,\mu^\prime_d)^T$ of $\vec{\mu}$ such that
	\begin{equation}
		|\mu_i-\mu^\prime_i| \le \frac{\sqrt{{\rm Tr}\Sigma}\log(d/\delta)}{n} \label{eq:err}
	\end{equation}
	for every $i\in[d]$ with probability at least $1-\delta$, making $\widetilde{O}(n)$ calls to $U_{\mathbb{P}}$ and $U_{\vec{\xi}}$.

	\label{th:MulEx}
\end{theorem}

On the oracles $U_{\mathbb{P}}$ and $U_{\vec{\xi}}$ needed for this algorithm, we will later present the concrete way to construct for CVaR contribution calculation.
Note that, along with $U_{\mathbb{P}}$ and $U_{\vec{\xi}}$, $\proc{QEstimator}_d(\vec{\xi},n,\delta)$ contains $\widetilde{O}(dn)$ uses of arithmetic circuits, since we calculate the inner product $\vec{u}\cdot\vec{\xi}$ with some $\vec{u}\in\mathbb{R}^d$ in the algorithm \cite{Cornelissen}.
However, we neglect them when we later count arithmetic computations in the proposed quantum method, since $U_{\mathbb{P}}$ and $U_{\vec{\xi}}$ in the method contain more arithmetic computations.
That is, as we will see later, they use arithmetic circuits $O(N_{\rm obl})$ times, and the total number in $\proc{QEstimator}_d(\vec{\xi},n,\delta)$ is therefore $O(nN_{\rm obl})$, which is larger than $\widetilde{O}(dn)$ since $d=N_{\rm gr}\le N_{\rm obl}$ in the current problem.


\subsection{\label{sec:QAA} Fixed-point quantum amplitude amplification}

We also explain the fixed-point QAA presented in \cite{Yoder}, which is an important subroutine in the proposed method.
This is a modified version of the original QAA algorithm in \cite{Brassard}.
Unlike the original one, the fixed-point QAA, as an unitary operation, amplifies the squared amplitude of the marked state in a superposition state to the value arbitrarily close to 1, without the unintended amplitude decay by too much iterative operations.
Formally, we have the following theorem.

\begin{theorem}[\cite{Yoder}]
	Let $\ket{\Phi}=\sqrt{p}\ket{\phi}\ket{1}+\sqrt{1-p}\ket{\phi^\prime}\ket{0}$ be a state on a system $R$ with a qubit $A$ attached, where $\ket{\phi}$ and $\ket{\phi^\prime}$ are states on $R$ and $p\in(0,1)$.
	Suppose that we are given an access to the oracle $U$ that prepares $\ket{\Phi}$: $U\ket{0}\ket{0}=\ket{\Phi}$.
	Then, for any $\delta\in(0,1)$, there is an unitary $V$ that prepares a state $\ket{\Phi^\prime}$ in the form of $\ket{\Phi^\prime}=\sqrt{1-\delta^\prime}\ket{\phi}\ket{1}+\sqrt{\delta^\prime}\ket{\phi^{\prime\prime}}\ket{0}$, where $\delta^\prime\in[0,\delta)$ and $\ket{\phi^{\prime\prime}}$ is some state on $R$, making $O\left(\frac{\log(\delta^{-1})}{\sqrt{p}}\right)$ calls to $U$.
	\label{th:FQAA}
\end{theorem}

Here, we assume that the marked state and other states are distinguished by whether a specific qubit $A$ is $\ket{1}$ or $\ket{0}$, which suffices for the current problem of CVaR contribution calculation.
\cite{Yoder} assumes that we can use the oracle $U_T$ such that, for the market state $\ket{T}$ with an ancilla qubit, $U_T\ket{T}\ket{b}=\ket{T}\ket{b\otimes1}$, and that, for any state $\ket{T^\prime}$ orthogonal to $\ket{T}$, $U_T\ket{T^\prime}\ket{b}=\ket{T^\prime}\ket{b}$, where $b\in\{0,1\}$.
Now, this is just a CNOT gate controlled by $A$, and therefore we will not consider the number of calls to it hereafter.

\subsection{Proposed algorithm}

We finally present the quantum algorithm for calculating CVaR contributions. 
Formally, we have the following theorem. 

\begin{theorem}
	Suppose that $N_{\rm obl}\in\mathbb{N}$, $e_1,...,e_{N_{\rm obl}}\in\mathbb{R}_+$, $a_1,...,a_{N_{\rm obl}}\in(0,1)$, $z_1,...,z_{N_{\rm obl}}\in\mathbb{R}$, $\alpha\in(0,1)$ and $\delta\in(0,1)$ are given.
	Suppose that $N_{\rm gr}\in\mathbb{N}$ and $n_1,...,n_{N_{\rm gr}}\in\mathbb{N}$ such that $\sum_{K=1}^{N_{\rm gr}}n_K=N_{\rm obl}$ are given.
	Let $v$ be a positive real number such that $p:={\rm Pr}(L\ge v)>0$, where $L$ is given in (\ref{eq:loss}).
	Suppose that we are given $C_{\rm max}\in\mathbb{R}$ satisfying
	\begin{equation}
	\max\{C^{1}_v(e_1,...,e_{N_{\rm obl}}),...,C^{N_{\rm gr}}_v(e_1,...,e_{N_{\rm obl}})\}\le C_{\rm max}. \label{eq:Cmax}
	\end{equation}
	Suppose that we are given $\sigma_{\rm max}\in\mathbb{R}$ satisfying
	\begin{equation}
		\max\{\sigma_1,...,\sigma_{N_{\rm gr}}\}\le \sigma_{\rm max}, \label{eq:sigmamax}
	\end{equation}
	where $(\sigma_K)^2$ is the conditional variance of the random variable $L_K$ in (\ref{eq:LK}) given that $L\ge v$.
	Suppose that we have an access to the SN state generation oracle $U^{\rm SN}$.
	Then, for any $\epsilon\in\mathbb{R}_+$ satisfying $\epsilon\le\sigma_{\rm max}\sqrt{N_{\rm gr}}$, there is a quantum algorithm that, with probability at least $1-\delta$, outputs the $\epsilon$-approximations $c_1,...,c_{N_{\rm gr}}\in\mathbb{R}$ of $C^{1}_v(e_1,...,e_{N_{\rm obl}}),...,C^{N_{\rm gr}}_v(e_1,...,e_{N_{\rm obl}})$, respectively, using $U^{\rm SN}$
	\begin{equation}
		\widetilde{O}\left(\frac{\sigma_{\rm max}\sqrt{N_{\rm gr}}}{\epsilon\sqrt{p}}\log\left(\max\left\{\frac{C_{\rm max}}{\epsilon},\frac{E_{\rm max}}{\sigma_{\rm max}}\right\}\right)\log\left(\frac{N_{\rm gr}}{\delta}\right)\right) \label{eq:callNumRN}
	\end{equation}
	times, and arithmetic circuits and angle-controlled Y rotations
	\begin{equation}
		\widetilde{O}\left(\frac{\sigma_{\rm max}N_{\rm obl}\sqrt{N_{\rm gr}}}{\epsilon\sqrt{p}}\log\left(\max\left\{\frac{C_{\rm max}}{\epsilon},\frac{E_{\rm max}}{\sigma_{\rm max}}\right\}\right)\log\left(\frac{N_{\rm gr}}{\delta}\right)\right) \label{eq:callNumAri}
	\end{equation}
	times.
	Here, $E_{\rm max}:=\max\{E_1,...,E_{N_{\rm gr}}\}$, where, for $K\in[N_{\rm gr}]$, $E_K:=\sum_{k=\kappa^K_1}^{\kappa^K_{n_K}}e_k$.
	\label{th:main}
\end{theorem}

\begin{proof}
	
	Consider the following probability space $(\Omega,2^\Omega,\mathbb{P})$.
	The sample space is
	\begin{equation}
		\Omega = [N_{\rm SN}]\times \{0,1\}^{N_{\rm obl}}.
	\end{equation}
	The probability measure $\mathbb{P}$ takes the form that, for each $\omega=(i,y_1,...,y_{N_{\rm obl}})\in\Omega$,
	\begin{equation}
		\mathbb{P}(\omega)=
		\begin{dcases}
		\frac{1-\epsilon^\prime}{p}\mathcal{P}(i;y_1,\cdots,y_{N_{\rm obl}}) & ; \ {\rm if} \ \sum_{k=1}^{N_{\rm obl}}e_ky_k\ge v \\
		\frac{\epsilon^\prime}{1-p}\widetilde{\mathbb{P}}(i,y_1,...,y_{N_{\rm obl}}) & ; \ {\rm otherwise}
		\end{dcases}.
	\end{equation}
	Here,
	\begin{eqnarray}
		&&\mathcal{P}(i;y_1,\cdots,y_{N_{\rm obl}}):= \nonumber \\
		&&\quad p^{\rm SN}_i\prod_{k=1}^{N_{\rm obl}}\left(P^{\rm def}_k(x^{\rm SN}_i)\delta_{y_k,1}+(1-P^{\rm def}_k(x^{\rm SN}_i))\delta_{y_k,0}\right)
	\end{eqnarray}
	is the probability that $X_0=x^{\rm SN}_i,Y_1=y_1,...,Y_{N_{\rm obl}}=y_{N_{\rm obl}}$, and $\frac{1}{p}\mathcal{P}(i;y_1,\cdots,y_{N_{\rm obl}})$ is the conditional probability of the same event given $\sum_{k=1}^{N_{\rm obl}}e_ky_k\ge v$.
	$\epsilon^\prime$ is some real number satisfying
	\begin{equation}
		0\le\epsilon^\prime\le\min\left\{\frac{\epsilon}{2C_{\rm max}},\left(\frac{\sigma_{\rm max}}{E_{\rm max}}\right)^2\right\}.
		\label{eq:epsPrCond}
	\end{equation}
	$\widetilde{\mathbb{P}}$ is some probability measure on a measurable space $(\Omega,2^{\Omega})$.
	Note that
	\begin{equation}
		\epsilon^\prime\le 1 \label{eq:epsprless1}
	\end{equation}
	holds because of $(\sigma_{\rm max})^2\le(E_{\rm max})^2$, which follows from $L_K\le E_{\rm max}$ for every $K\in[N_{\rm gr}]$, and that (\ref{eq:epsprless1}) makes $\mathbb{P}$ non-negative and therefore well-defined.
	Besides, consider the following random variable $\vec{\xi}=(\xi_1,...,\xi_{N_{\rm gr}})^T\in\mathbb{R}^{N_{\rm gr}}$ on $(\Omega,2^\Omega,\mathbb{P})$: for any $K\in[N_{\rm gr}]$ and $\omega=(i,y_1,...,y_{N_{\rm obl}},w)\in\Omega$,
	\begin{equation}
		\xi_K(\omega) := 
		\begin{cases}
			\sum_{k=\kappa^K_1}^{\kappa^K_{n_K}}e_ky_k  & ; \ {\rm if} \ \sum_{k=1}^{N_{\rm obl}}e_ky_k\ge v \\
			0 & ; \ {\rm otherwise}
		\end{cases}. \label{eq:xiK}
	\end{equation}
	Then, the expected value of $\xi_K$ is
	\begin{equation}
		\mathbb{E}[\xi_K]=(1-\epsilon^\prime)C^K_v(e_1,...,e_{N_{\rm obl}}).
	\end{equation}
	Because of (\ref{eq:Cmax}) and (\ref{eq:epsPrCond}),
	\begin{equation}
		\left|C^K_v(e_1,...,e_{N_{\rm obl}})-\mathbb{E}[\xi_K]\right| \le \frac{\epsilon}{2}
	\end{equation}
	holds.
	Therefore, if we obtain $\epsilon/2$-approximations $\bar{\xi}_1,...,\bar{\xi}_{N_{\rm gr}}$ of $\mathbb{E}[\xi_1],...,\mathbb{E}[\xi_{N_{\rm gr}}]$, they are also $\epsilon$-approximations of $C^1_v,...,C^{N_{\rm gr}}_v$.

	Thus, we hereafter consider how to obtain such $\bar{\xi}_1,...,\bar{\xi}_{N_{\rm gr}}$ by $\proc{QEstimator}$.
	To use this, we need the oracles $U_{\mathbb{P}}$ in (\ref{eq:UP}) and $U_{\vec{\xi}}$ in (\ref{eq:Uxi}).
	First, let us consider $U_{\mathbb{P}}$.
	Define
	\begin{eqnarray}
		&&\ket{\Psi_{\ge v}}:= \nonumber \\
		&& \quad \frac{1}{\sqrt{p}}\sum_{i=1}^{N_{\rm SN}}\sum_{\substack{y_1,...,y_{N_{\rm obl}}\in\{0,1\} \\ \sum_{k=1}^{N_{\rm obl}}e_ky_k\ge v}}\sqrt{\mathcal{P}(i;y_1,\cdots,y_{N_{\rm obl}})}\ket{x^{\rm SN}_i}\ket{y_1}\cdots \ket{y_{N_{\rm obl}}}, \nonumber \\
		&&
	\end{eqnarray}
	as a quantum state on a system $S$ consisting of a quantum register and $N_{\rm obl}$ qubits.
	On the system consisting of $S$ and an additional qubit, we can perform the following operation:
	\begin{eqnarray}
		&& \ket{0}\ket{0}^{\otimes N_{\rm obl}}\ket{0} \nonumber \\
		&\rightarrow&\sum_{i=1}^{N_{\rm SN}}\sqrt{p^{\rm SN}_i}\ket{x^{\rm SN}_i}\ket{0}^{\otimes N_{\rm obl}}\ket{0} \nonumber \\
		&\rightarrow&\sum_{i=1}^{N_{\rm SN}}\sqrt{p^{\rm SN}_i}\ket{x^{\rm SN}_i} \left(\sqrt{P^{\rm def}_1(x^{\rm SN}_i)}\ket{1}+\sqrt{1-P^{\rm def}_1(x^{\rm SN}_i)}\ket{0}\right) \nonumber \\
		&&\qquad\quad \otimes\cdots\otimes\left(\sqrt{P^{\rm def}_{N_{\rm obl}}(x^{\rm SN}_i)}\ket{1}+\sqrt{1-P^{\rm def}_{N_{\rm obl}}(x^{\rm SN}_i)}\ket{0}\right)\ket{0} \nonumber \\
		&=&\sum_{i=1}^{N_{\rm SN}}\sum_{y_1,...,y_{N_{\rm obl}}\in\{0,1\}}\sqrt{\mathcal{P}(i;y_1,\cdots,y_{N_{\rm obl}})}\ket{x^{\rm SN}_i}\ket{y_1}\cdots\ket{y_{N_{\rm obl}}}\ket{0} \nonumber \\
		&\rightarrow& \sum_{i=1}^{N_{\rm SN}}\sum_{y_1,...,y_{N_{\rm obl}}\in\{0,1\}}\sqrt{\mathcal{P}(i;y_1,\cdots,y_{N_{\rm obl}})}\ket{x^{\rm SN}_i}\ket{y_1}\cdots\ket{y_{N_{\rm obl}}} \nonumber \\
		&& \qquad\qquad\qquad\qquad \otimes \left(1_{\sum_{k=1}^{N_{\rm obl}}e_ky_k\ge v}\ket{1}+1_{\sum_{k=1}^{N_{\rm obl}}e_ky_k< v}\ket{0}\right) \nonumber \\
		&=& \sqrt{p}\ket{\Psi_{\ge v}}\ket{1}+\sqrt{1-p}\ket{\Psi_{\rm gar}}\ket{0}, \label{eq:opgev}
	\end{eqnarray}
	where $\ket{\Psi_{\rm gar}}$ is some state on $S$ and some ancillary registers are not displayed.
	In (\ref{eq:opgev}), we use $U^{\rm SN}$ at the first arrow.
	At the second arrows, we calculate $\arccos\left(\sqrt{P^{\rm def}_1(x^{\rm SN}_i)}\right),...,\arccos\left(\sqrt{P^{\rm def}_{N_{\rm obl}}(x^{\rm SN}_i)}\right)$ onto undisplayed ancillary registers using arithmetic circuits, and then operate $U_{\rm CRY}$ with the rotation angles specified by these ancillary registers.
	At the third arrow, we calculate $\sum_{k=1}^{N_{\rm obl}}e_ky_k$ by $N_{\rm obl}$ additions and multiplications, then use $U_{\rm comp}$.
	We denote by $U_{\ge v}$ the oracle that acts as (\ref{eq:opgev}).
	Then, because of Theorem \ref{th:FQAA}, we can generate a state
	\begin{equation}
		\ket{\Phi_{\ge v}}:=\sqrt{1-\epsilon^\prime}\ket{\Psi_{\ge v}}\ket{1}+\sqrt{\epsilon^\prime}\ket{\Psi_{\rm gar}^\prime}\ket{0}, \label{eq:Phiv}
	\end{equation}
	where $\ket{\Psi_{\rm gar}^\prime}$ is some state and $\epsilon^\prime$ is some real number satisfying (\ref{eq:epsPrCond}), making
	\begin{equation}
		O\left(\frac{\log\left(\max\left\{\frac{C_{\rm max}}{\epsilon},\frac{E_{\rm max}}{\sigma_{\rm max}}\right\}\right)}{\sqrt{p}}\right) \label{eq:numUgev}
	\end{equation}
	calls to $U_{\ge v}$.
	Note that $\ket{\Phi_{\ge v}}$ is in fact $\sum_{\omega\in\Omega} \sqrt{\mathbb{P}(\omega)}\ket{\omega}$ with identification that
	\begin{equation}
		\ket{\omega} =
		\begin{cases}
			\ket{x^{\rm SN}_i}\ket{y_1}\cdots\ket{y_{N_{\rm obl}}}\ket{1} & ; \ {\rm if} \ \sum_{k=1}^{N_{\rm obl}} e_ky_k\ge v \\
			\ket{x^{\rm SN}_i}\ket{y_1}\cdots\ket{y_{N_{\rm obl}}}\ket{0} & ; \ {\rm otherwise}
		\end{cases}
		\label{eq:ketomega}
	\end{equation}
	for each $\omega=(i,y_1,...,y_{N_{\rm obl}})\in\Omega$.
	Therefore, we hereafter denote by $U_{\mathbb{P}}$ the oracle that generates $\ket{\Phi_{\ge v}}$ from $\ket{0}\ket{0}^{\otimes N_{\rm obl}}\ket{0}$.
	
	On the other hand, constructing $U_{\vec{\xi}}$ is more simple.
	Given a state $\ket{\omega}$ in (\ref{eq:ketomega}) for $\omega\in\Omega$, we can add $N_{\rm gr}$ registers and perform the operation
	\begin{eqnarray}
		&&\ket{x^{\rm SN}_i}\ket{y_1}\cdots\ket{y_{N_{\rm obl}}}\ket{w}\ket{0}^{N_{\rm gr}} \nonumber \\
		&\rightarrow& \ket{x^{\rm SN}_i}\ket{y_1}\cdots\ket{y_{N_{\rm obl}}}\ket{w}\bigotimes_{K=1}^{N_{\rm gr}}\Ket{w\sum_{k=\kappa^K_1}^{\kappa^K_{n_K}}e_ky_k}, \label{eq:UxiConc}
	\end{eqnarray}
	where $w\in\{0,1\}$.
	This can be done by $O(N_{\rm obl})$ multiplications and additions.
	Note that (\ref{eq:UxiConc}) is in fact $U_{\vec{\xi}}$ for $\vec{\xi}$ in (\ref{eq:xiK}), since $w=1$ is $\sum_{k=1}^{N_{\rm obl}} e_ky_k\ge v$ and 0 otherwise.
	
	With $U_{\mathbb{P}}$ and $U_{\vec{\xi}}$ constructed as above, we perform $\proc{QEstimator}_{N_{\rm gr}}\left(\vec{\xi},n,\delta\right)$ with
	\begin{equation}
		n=\left\lceil\frac{2\sqrt{2}\sigma_{\max}\sqrt{N_{\rm gr}}\log(N_{\rm gr}/\delta)}{\epsilon}\right\rceil.
	\end{equation}
	Because of (\ref{eq:err}), each of outcomes $\bar{\xi}_1,...,\bar{\xi}_{N_{\rm gr}}$ of this satisfies 
	\begin{eqnarray}
		\left|\bar{\xi}_K-\mathbb{E}[\xi_K]\right|&\le& \frac{\sqrt{{\rm Tr}\widetilde{\Sigma}}\log(N_{\rm gr}/\delta)}{n} \nonumber \\
		&\le& \frac{\sqrt{N_{\rm gr}(\widetilde{\sigma}_{\max})^2}\log(N_{\rm gr}/\delta)}{n} \nonumber \\
		&\le& \frac{\sqrt{2N_{\rm gr}(\sigma_{\max})^2}\log(N_{\rm gr}/\delta)}{n} \nonumber \\
		&\le& \frac{\epsilon}{2}, \label{eq:errFin}
	\end{eqnarray}
	where $\widetilde{\Sigma}$ is the covariance matrix of $\xi_1,...,\xi_{N_{\rm gr}}$, and $(\widetilde{\sigma}_{\max})^2$ is its largest diagonal element, that is, the maximum of the variances of $\xi_1,...,\xi_{N_{\rm gr}}$.
	The third inequality in (\ref{eq:errFin}) holds since
	\begin{equation}
		(\widetilde{\sigma}_{\max})^2 \le 2(\sigma_{\max})^2, \label{eq:sigmatilsigma}
	\end{equation}
	whose proof is postponed to Appendix \ref{app}.
	(\ref{eq:errFin}) means that $\bar{\xi}_1,...,\bar{\xi}_{N_{\rm gr}}$ are $\epsilon/2$-approximations of $\mathbb{E}[\xi_1],...,\mathbb{E}[\xi_{N_{\rm gr}}]$, and also $\epsilon$-approximations of $C^1_v,...,C^{N_{\rm gr}}_v$, as discussed above. 

	Finally, let us count the numbers of calls to building-block circuits in this algorithm. $\proc{QEstimator}_{N_{\rm gr}}\left(\vec{\xi},n,\delta\right)$ calls $U_{\mathbb{P}}$ and $U_{\vec{\xi}}$ $\widetilde{O}(n)$ times, that is,
	\begin{equation}
		\widetilde{O}\left(\frac{\sigma_{\max}\sqrt{N_{\rm gr}}\log(N_{\rm gr}/\delta)}{\epsilon}\right)
	\end{equation}
	times.
	$U_{\mathbb{P}}$ makes iterative calls to $U_{\ge v}$, whose number is evaluated as (\ref{eq:numUgev}).
	In $U_{\ge v}$, $U^{\rm SN}$ is called once, and arithmetic circuits and $U_{\rm CRY}$ are called $O(N_{\rm obl})$ times.
	On the other hand, $U_{\vec{\xi}}$ consists of $O(N_{\rm obl})$ arithmetic circuits.
	In total, for the numbers of calls to building-block circuits, we obtain evaluations (\ref{eq:callNumRN}) and (\ref{eq:callNumAri}).
\end{proof}

\subsection{Comparison with the classical Monte Carlo method}

We have seen that the proposed quantum algorithm estimates $N_{\rm gr}$ CVaR contributions with accuracy $\epsilon$ and has query complexity scaling on $N_{\rm gr}$ and $\epsilon$ as $O(\sqrt{N_{\rm gr}}/\epsilon)$.
Seemingly, this means so-called quadratic quantum speedup with respect to both $N_{\rm gr}$ and $\epsilon$. 
However, if we do not require the scaling on $\epsilon$ to be $O(\epsilon^{-1})$, we can achieve the better dependence on $N_{\rm gr}$ classically.
In classical Monte Carlo integration, we generate many samples of $X_0$ and $Y_1,...,Y_{N_{\rm gr}}$, calculate $\widetilde{L}_1,...,\widetilde{L}_{N_{\rm gr}}$ for these samples, and take averages.
By this procedure, we obtain $\epsilon$-approximations for $N_{\rm gr}$ CVaR contributions with high probability with $\widetilde{O}(\log N_{\rm gr}/\epsilon^{-2})$ sample complexity \cite{Cornelissen}.
This implies that, depending on $N_{\rm gr}$ and $\epsilon$, the quantum method might not be the best way.

In order to identify the situation where the quantum or classical method is better more precisely, let us evaluate the complexity in the classical method in more detail.
The sample complexity to obtain $\epsilon$-approximations for $C^1_v,...,C^{N_{\rm gr}}_v$ with probability at least $\delta$ in the classical method is \cite{Cornelissen}
\begin{equation}
	\widetilde{\Theta}\left(\frac{(\sigma_{\rm max})^2\log (N_{\rm gr}/\delta)}{\epsilon^2p}\right). \label{eq:clComp}
\end{equation}
Here, the factor $1/p$ appears since we randomly generate $\Theta(N/p)$ samples of $X_0,Y_1,...,Y_{N_{\rm obl}}$ in order to obtain $N$ samples such that $L\ge v$.
In one sample generation, one standard normal random variable is generated, $\Theta(N_{\rm obl})$ arithmetic operations are done, and $N_{\rm obl}$ Bernoulli random variables are generated.
In total, in the classical method, the number of standard normal random variable generations is evaluated as (\ref{eq:clComp}), and the numbers of arithmetic operations and Bernoulli random variable generations are 
\begin{equation}
	\widetilde{\Theta}\left(\frac{(\sigma_{\rm max})^2N_{\rm obl}\log (N_{\rm gr}/\delta)}{\epsilon^2p}\right).
\end{equation}
Then, let us compare these complexity estimations with those for the proposed quantum method.
We see that the quantum method is advantageous in query complexity if
\begin{equation}
	\frac{(\sigma_{\rm max})^2}{\epsilon^2p}\gg N_{\rm gr}, \label{eq:condQAdv}
\end{equation}
although this is a rough discussion with logarithmic factors neglected.

As an extreme case, let us assume that we want CVaR contributions for all the individual obligors.
In this case, $N_{\rm gr}$ takes the maximum $N_{\rm obl}$, and therefore it is most disadvantageous for the quantum method.
Besides, we roughly evaluate $(\sigma_{\rm max})^2$ as $\bar{p}_{\rm def}(1-\bar{p}_{\rm def})\bar{e}^2$, where $\bar{p}_{\rm def}$ and $\bar{e}$ are the typical scale of conditional default probabilities of obligors given $L\ge v$ and that of exposures, respectively.
Moreover, we set the accuracy $\epsilon$ relative to $C_{\rm max}$ and roughly evaluate $C_{\rm max}$ as $\bar{p}_{\rm def}\bar{e}$.
With these evaluations, (\ref{eq:condQAdv}) becomes
\begin{equation}
	\left(\frac{C_{\rm max}}{\epsilon}\right)^2\frac{1-\bar{p}_{\rm def}}{p\bar{p}_{\rm def}}\gg N_{\rm obl}. \label{eq:condQAdv2}
\end{equation}
Under a plausible assumption that $p=0.01,p_{\rm def}\sim 0.01,C_{\rm max}/{\epsilon}\sim0.01$, the left hand side of (\ref{eq:condQAdv2}) is of order $10^8$, and therefore the quantum method is promising to be advantageous when $N_{\rm obl}\ll 10^8$, which is usually satisfied.
In the case that $N_{\rm gr}\ll N_{\rm obl}$, which means that we need CVaR contributions not for individual obligors but for some obligor groups, the quantum method becomes more advantageous. 

\section{\label{sec:sum}Summary}

In this paper, we considered application of quantum algorithms to CVaR contribution calculation, which is an important problem in financial risk management but has not been focused in the context of quantum computing.
We clarified how to apply the recent quantum algorithm for multiple expected value estimation \cite{Cornelissen} to this problem as estimation of many conditional expected value, constructing the oracles needed in the algorithm with fixed-point QAA \cite{Yoder}.
We then evaluated the numbers of queries to the building-block oracles, arithmetic operations, controlled rotation and SN state generation, as (\ref{eq:callNumRN}) and (\ref{eq:callNumAri}).
In terms of query complexity, the proposed method achieves quantum quadratic speedup with respect to the accuracy $\epsilon$, but the scaling of complexity on the number of obligor groups $N_{\rm gr}$ is better in the classical method.
This means that, for small $\epsilon$ and $N_{\rm gr}$, the proposed method becomes advantageous against the classical method.
We saw that, even when we calculate CVaR contributions for all the individual obligors, which is most disadvantageous for the quantum method, its query complexity is smaller than the classical method in some plausible setting.

But, unfortunately, it is obvious that the need for risk contributions of many finely divided obligor groups, which enable a bank to analyze the risk in details, reduces the quantum advantage compared with calculating only the risk measures for the entire portfolio, which was originally considered in \cite{Egger2}.
As future works, not only for credit risk management but also for other types of financial problems, we should extend discussion on quantum computing application to various problem settings arising in practice and scrutinize the degree of quantum advantage.

\section*{Acknowledgment}

This work was supported by MEXT Quantum Leap Flagship Program (MEXT Q-LEAP) Grant Number JPMXS0120319794.

\appendix

\section{Proof of (\ref{eq:sigmatilsigma}) \label{app}}

Let $\mathbb{P}_0$ be a probability measure on the measurable space $(\Omega,2^\Omega)$ such that
\begin{equation}
	\mathbb{P}_0(\omega)=
	\begin{dcases}
		\frac{1}{p}\mathcal{P}(i;y_1,\cdots,y_{N_{\rm obl}}) & ; \ {\rm if} \ \sum_{k=1}^{N_{\rm obl}}e_ky_k\ge v \\
		0 & ; \ {\rm otherwise}
	\end{dcases},
\end{equation}
for each $\omega=(i,y_1,...,y_{N_{\rm obl}})\in\Omega$.
For each $K\in[N_{\rm gr}]$, the variances of $\xi_K$ with respect to $\mathbb{P}_0$ is equal to $(\sigma_K)^2$, and written as
\begin{equation}
	(\sigma_K)^2 = \sum_{\omega\in\Omega_{\ge v}} \mathbb{P}_0(\omega) (\xi_K(\omega))^2 - \left(\sum_{\omega\in\Omega_{\ge v}} \mathbb{P}_0(\omega) \xi_K(\omega)\right)^2, \label{eq:sigmaK}
\end{equation} 
where $\Omega_{\ge v}:=\left\{(i,y_1,...,y_{N_{\rm obl}})\in\Omega \ \middle| \ \sum_{k=1}^{N_{\rm obl}}e_ky_k\ge v\right\}$.
On the other hand, since $\mathbb{P}(\omega)=(1-\epsilon^\prime)\mathbb{P}_0(\omega)$ for $\omega\in\Omega_{\ge v}$, we have
\begin{eqnarray}
	&&(\widetilde{\sigma}_K)^2 \nonumber \\
	&=& \sum_{\omega\in\Omega_{\ge v}} \mathbb{P}(\omega) (\xi_K(\omega))^2 - \left(\sum_{\omega\in\Omega_{\ge v}} \mathbb{P}(\omega) \xi_K(\omega)\right)^2  \nonumber \\
	&=& (1-\epsilon^\prime)\sum_{\omega\in\Omega_{\ge v}} \mathbb{P}_0(\omega) (\xi_K(\omega))^2 - (1-\epsilon^\prime)^2\left(\sum_{\omega\in\Omega_{\ge v}} \mathbb{P}_0(\omega) \xi_K(\omega)\right)^2 \nonumber \\
	&=& (\sigma_K)^2 -\epsilon^\prime\sum_{\omega\in\Omega_{\ge v}} \mathbb{P}_0(\omega) (\xi_K(\omega))^2+\epsilon^\prime(2-\epsilon^\prime)\left(\sum_{\omega\in\Omega_{\ge v}} \mathbb{P}_0(\omega) \xi_K(\omega)\right)^2 \nonumber \\
	&\le& (\sigma_K)^2 +\epsilon^\prime(1-\epsilon^\prime)\sum_{\omega\in\Omega_{\ge v}} \mathbb{P}_0(\omega) (\xi_K(\omega))^2. \nonumber \\
	&& \label{eq:temp}
\end{eqnarray}
In (\ref{eq:temp}), at the inequality, we use (\ref{eq:epsprless1}) and $\sum_{\omega\in\Omega_{\ge v}} \mathbb{P}_0(\omega) (\xi_K(\omega))^2 \ge \left(\sum_{\omega\in\Omega_{\ge v}} \mathbb{P}_0(\omega) \xi_K(\omega)\right)^2$, which follows from the fact that the variance $(\sigma_K)^2$ is non-negative.
Then, applying $\xi_K(\omega)\le E_{\rm max}$, (\ref{eq:epsPrCond}) and (\ref{eq:epsprless1}) to (\ref{eq:temp}), we obtain
\begin{equation}
	(\widetilde{\sigma}_K)^2 \le (\sigma_K)^2 + (\sigma_{\rm max})^2,
\end{equation}
which implies $(\widetilde{\sigma}_{\rm max})^2 \le 2(\sigma_{\rm max})^2$.

\end{document}